\documentclass[11pt]{article}

\usepackage{bbm}
\usepackage{amssymb}
\usepackage{amsmath}
\usepackage{amsthm}
\usepackage{color}
\usepackage{colortbl}
\usepackage{graphicx}
\usepackage{hyperref}
\usepackage{fullpage}
\usepackage{xspace}
\usepackage{subfigure}

\usepackage{lmodern}
\usepackage[T1]{fontenc}

\newcommand{\R}{\mathbb{R}}

\DeclareMathOperator*{\E}{\mathbb{E}}

\newtheorem{theorem}{Theorem}
\newtheorem{lemma}[theorem]{Lemma}

\newcommand{\eps}{\epsilon}

\newcommand{\etal}{\emph{et al.}\xspace}

\usepackage{caption}
\usepackage{algorithm}
\usepackage{algpseudocode}

\newcommand{\LineComment}[1]{\textcolor{blue}{$\langle\langle$  #1 $\rangle\rangle$ } }

\newcommand{\zero}{\vec{0}}
\newcommand{\one}{\vec{1}}
\newcommand{\xopt}{\vec{x}^*}

\newcommand{\xstart}{\vec{x}_{\mathrm{start}}}
\newcommand{\xend}{\vec{x}_{\mathrm{end}}}

\newcommand{\zstart}{\vec{z}_{\mathrm{start}}}
\newcommand{\zend}{\vec{z}_{\mathrm{end}}}

\newcommand{\vstart}{v_{\mathrm{start}}}
\newcommand{\vend}{v_{\mathrm{end}}}

\newcommand{\x}{\vec{x}}
\newcommand{\z}{\vec{z}}
\newcommand{\g}{\vec{g}}
\renewcommand{\u}{\vec{u}}
\renewcommand{\v}{\vec{v}}
\newcommand{\w}{\vec{w}}

\renewcommand{\a}{\vec{a}}
\renewcommand{\b}{\vec{b}}
\renewcommand{\c}{\vec{c}}
\renewcommand{\d}{\vec{d}}
\newcommand{\m}{\vec{m}}

\newcommand{\zet}{\vec{z}(\eta)}
\newcommand{\get}{\vec{g}(\eta)}
\newcommand{\Set}{S(\eta)}
\newcommand{\Tet}{T(\eta)}
\newcommand{\xp}{\vec{x}'}
\newcommand{\zp}{\vec{z}'}

\allowdisplaybreaks
\begin{document}

\title{Parallel Algorithm for Non-Monotone DR-Submodular Maximization}
\author{
Alina Ene\thanks{Department of Computer Science, Boston University, {\tt aene@bu.edu}.}
\and
Huy L. Nguy\~{\^{e}}n\thanks{College of Computer and Information Science, Northeastern University, {\tt hlnguyen@cs.princeton.edu}.} 
}
\date{}

\maketitle

\begin{abstract}
In this work, we give a new parallel algorithm for the problem of maximizing a non-monotone diminishing returns submodular function subject to a cardinality constraint. For any desired accuracy $\epsilon$, our algorithm achieves a $1/e - \epsilon$ approximation using $O(\log{n} \log(1/\epsilon) / \epsilon^3)$ parallel rounds of function evaluations. The approximation guarantee nearly matches the best approximation guarantee known for the problem in the sequential setting and the number of parallel rounds is nearly-optimal for any constant $\epsilon$. Previous algorithms achieve worse approximation guarantees using $\Omega(\log^2{n})$ parallel rounds. Our experimental evaluation suggests that our algorithm obtains solutions whose objective value nearly matches the value obtained by the state of the art sequential algorithms, and it outperforms previous parallel algorithms in number of parallel rounds, iterations, and solution quality.
\end{abstract}

\section{Introduction}

In this paper, we study parallel algorithms for the problem of maximizing a non-monotone DR-submodular function subject to a single cardinality constraint\footnote{A DR-submodular function $f$ is a continuous function with the diminishing returns property: if $x\le y$ coordinate-wise then $\nabla f(x) \ge \nabla f(y)$ coordinate-wise.}. The problem is a generalization of submodular maximization subject to a cardinality constraint. Many recent works have shown that DR-submodular maximization has a wide-range of applications beyond submodular maximization. These applications include maximum a-posteriori (MAP) inference for determinantal point processes (DPP), mean-field inference in log-submodular models, quadratic programming, and revenue maximization in social networks~\cite{kulesza2012determinantal,gillenwater2012near,bian2016guaranteed,ito2016large,soma2017non,bian2017continuous,bian2018optimal}. 

The problem of maximizing a DR-submodular function subject to a convex constraint is a notable example of a \emph{non-convex} optimization problem that can be solved with \emph{provable} approximation guarantees. The continuous Greedy algorithm \cite{Vondrak2008} developed in the context of the multilinear relaxation framework  applies more generally to maximizing DR-submodular functions that are monotone increasing (if $\vec{x} \leq \vec{y}$ coordinate-wise then $f(\vec{x}) \leq f(\vec{y})$). Chekuri \etal \cite{ChekuriJV15} developed algorithms for both monotone and non-monotone DR-submodular maximization subject to packing constraints that are based on the continuous Greedy and multiplicative weights update framework. The work~\cite{bian2017continuous} generalized continuous Greedy for submodular functions to the DR-submodular case and developed Frank-Wolfe-style algorithms for maximizing non-monotone DR-submodular function subject to general convex constraints.

A significant drawback of these algorithms is that they are inherently sequential and adaptive. In fact the highly adaptive nature of these algorithms go back to the classical greedy algorithm for submodular functions: the algorithm sequentially selects the next element based on the marginal gain on top of previous elements. In certain settings such as feature selection~\cite{KEDNG17} evaluating the objective function is a time-consuming procedure and the main bottleneck of the optimization algorithm and therefore, parallelization is a must. Recent lines of work have focused on addressing these shortcomings and understanding the trade-offs between approximation guarantee, parallelization, and adaptivity. Starting with the work of Balkanski and Singer~\cite{BS18}, there have been very recent efforts to understand the tradeoff between approximation guarantee and adaptivity for submodular maximization~\cite{BS18,EN18,BRS18,FMZ18,CQ18,BBY18}. The \emph{adaptivity} of an algorithm is the number of sequential rounds of queries it makes to the evaluation oracle of the function, where in every round the algorithm is allowed to make polynomially-many parallel queries. Recently, the work~\cite{fahrbach2018non} gave an algorithm for maximizing a submodular function subject to a cardinality constraint in $O(\log n/\eps)$ rounds and $0.031-\eps$ approximation. For the general setting of DR-submodular functions with $m$ packing constraints, the work~\cite{ENV19} gave an algorithm with $O(\log (n/\eps)\log(1/\eps) \log (m+n)/\eps^2)$ rounds and $1/e-\eps$ approximation. In the special case of $m=1$ constraint, this algorithm uses $O(\log^2 n / \eps^2)$ rounds.

In this work, we develop a new algorithm for DR-submodular maximization subject to a single cardinality constraint using $O(\log n\log(1/\eps) / \eps^3)$ rounds of adaptivity and obtaining $1/e-\eps$ approximation. For constant $\eps$, the number of rounds is almost a quadratic improvement from $O(\log^2 n)$ in the previous work to the nearly optimal $O(\log n)$ rounds.  

\begin{theorem}
Let $f: [0, 1]^n \rightarrow \R_+$ be a DR-submodular function and $k \in \R_+$. For every $\eps > 0$, there is an algorithm for the problem $\max_{\x \in [0, 1]^n \colon \|\x\|_1 \leq k} f(\x)$ with the following guarantees:
\begin{itemize}
\item The algorithm is deterministic if provided oracle access for evaluating $f$ and its gradient $\nabla f$;
\item The algorithm achieves an approximation guarantee of $\frac{1}{e} -\eps$;
\item The number of rounds of adaptivity is $O\left(\frac{\log n\log(1/\eps)}{\eps^3} \right)$.
\end{itemize}
\end{theorem}

\section{Preliminaries}
In this paper, we consider non-negative functions $f: [0, 1]^n \rightarrow \R_+$ that are \emph{diminishing returns submodular} (DR-submodular). A function is DR-submodular if $\forall \vec{x} \leq \vec{y} \in [0, 1]^n$ (where $\leq$ is coordinate-wise), $\forall i \in [n]$, $\forall \delta \in [0, 1]$ such that $\vec{x} + \delta \vec{1}_i$ and $\vec{y} + \delta \vec{1}_i$ are still in $[0, 1]^n$, it holds
  \[f(\vec{x} + \delta \vec{1}_i) - f(\vec{x}) \geq f(\vec{y} + \delta \vec{1}_i) - f(\vec{y}),\]
where $\vec{1}_i$ is the $i$-th basis vector, i.e., the vector whose $i$-th entry is $1$ and all other entries are $0$.

If $f$ is differentiable, $f$ is DR-submodular if and only if $\nabla f(\vec{x}) \geq \nabla f(\vec{y})$ for all $\vec{x} \leq \vec{y} \in [0, 1]^n$. If $f$ is twice-differentiable, $f$ is DR-submodular if and only if all the entries of the Hessian are \emph{non-positive}, i.e., $\frac{\partial^2 f}{\partial x_i \partial x_j}(\vec{x}) \leq 0$ for all $i, j \in [n]$. 

For simplicity, throughout the paper, we assume that $f$ is differentiable. We assume that we are given black-box access to an oracle for evaluating $f$ and its gradient $\nabla f$. It is convenient to extend the function $f$ to $\R^n_+$ as follows: $f(\vec{x}) = f(\vec{x} \wedge \vec{1})$, where $(\vec{x} \wedge \vec{1})_i = \min\{x_i, 1\}$.

An example of a DR-submodular function is the multilinear extension of a submodular function. The multilinear extension $f: [0, 1]^V \rightarrow \R$ of a submodular function $F: 2^V \rightarrow \R$ is defined as follows:
  \[ f(\vec{x}) = \E[F(R(\vec{x}))] = \sum_{S \subseteq V} F(S) \prod_{i \in S} \vec{x}_i \prod_{i \in V \setminus S} (1 - \vec{x}_i),\]
where $R(\vec{x})$ is a random subset of $V$ where each $i \in V$ is included independently at random with probability $\vec{x}_i$.

\medskip
{\bf Basic notation.} We use e.g. $\vec{x} = (\vec{x}_1, \dots, \vec{x}_n)$ to denote a vector in $\R^n$. We use the following vector operations: $\vec{x} \vee \vec{y}$ is the vector whose $i$-th coordinate is $\max\{x_i, y_i\}$; $\vec{x} \wedge \vec{y}$ is the vector whose $i$-th coordinate is $\min\{x_i, y_i\}$; $\vec{x} \circ \vec{y}$ is the vector whose $i$-th coordinate is $x_i \cdot y_i$. We write $\vec{x} \leq \vec{y}$ to denote that $\vec{x}_i \leq \vec{y}_i$ for all $i \in [n]$. Let $\vec{0}$ (resp. $\vec{1}$) be the $n$-dimensional all-zeros (resp. all-ones) vector. Let $\vec{1}_S \in \{0, 1\}^V$ denote the indicator vector of $S \subseteq V$, i.e., the vector that has a $1$ in entry $i$ if and only if $i \in S$.

We will use the following result that was shown in previous work~\cite{ChekuriJV15}.

\begin{lemma}[\cite{ChekuriJV15}, Lemma~7]
\label{lem:x-or-opt}
Let $f: [0, 1]^n \rightarrow \R_+$ be a DR-submodular function. For all $\vec{x}^* \in [0, 1]^n$ and $\vec{x} \in [0, 1]^n$, $f(\vec{x}^* \vee \vec{x}) \geq (1 - \|\vec{x}\|_{\infty}) f(\vec{x}^*)$.
\end{lemma}

\section{The algorithm}
\label{sec:algo}

In this section, we present an idealized version of our algorithm where we assume that we can compute exactly the step size on line \ref{line:eta1}. The idealized algorithm is given in Algorithm~\ref{alg:non-monotone-ideal}. In the appendix (Section~\ref{app:approx-step-size}), we show how to implement that step efficiently and incur only $O(\eps)$ additive error in the approximation.

The algorithm takes as input a target value $M$ and it achieves the desired $1/e - O(\eps)$ approximation if $M$ is an $(1 + \eps)$ approximation of the optimal function value $f(\xopt)$, i.e., we have $f(\xopt) \leq M \leq (1 + \eps) f(\xopt)$. As noted in previous work~\cite{ENV19}, it is straightforward to approximately guess such a value $M$ using a single parallel round.

\begin{algorithm}
\caption{Algorithm for $\max_{\vec{x} \in [0, 1]^n \colon \|\vec{x}\|_1 \leq k} f(\vec{x})$, where $f$ is a non-negative DR-submodular function. The algorithm takes as input a target value $M$ such that $f(\xopt) \leq M \leq (1 + \eps) f(\xopt)$.} 
\label{alg:non-monotone-ideal}
\begin{algorithmic}[1]
\State $\x \gets \zero$
\State $\z \gets \zero$
\For{$j = 1$ to $1/\eps$}
  \State \LineComment{Start of phase $j$}
  \State $\xstart \gets \x$
  \State $\zstart \gets \z$
  \State $\vstart \gets \frac{1}{k}(((1 - \eps)^j - 2\eps) M - f(\x))$ \label{line:vstart}
  \State $v \gets \vstart$
  \While{$v > \eps \vstart$ and $\|\z\|_1 < \eps jk$}
    \State $\g = (\one - \z) \circ \nabla f(\z)$
    \State $S = \{i \in [n] \colon \g_i \geq v \text{ and } \z_i \leq 1 - (1 - \eps)^j \text{ and } \z_i - (\zstart)_i < \eps (1 - (\zstart)_i)\}$ \label{line:S}
    \If{$S = \emptyset$}
      \State $v \gets (1 - \eps) v$ \label{line:update-v}
    \Else
      \State For a given $\eta \in [0, \eps^2]$, we define:
      \begin{align*}
        \zet &= \z + \eta (\one - \z) \circ \one_S\\
        \get &= (\one - \zet) \circ \nabla f(\zet)\\
        \Set &= \{i \in S \colon \get_i \geq v \}\\
        \Tet &= \{i \in S \colon \get_i > 0 \}
      \end{align*}
      \label{line:defs}
      \State Let $\eta_1$ be the maximum $\eta \in [0, \eps^2]$ such that $|\Set| \geq (1 - \eps) |S|$ \label{line:eta1}
      \State \LineComment{$\eta_2 = \min\left\{\eps^2,  \frac{\eps jk - \|\z\|_1}{|S| - \|\z \circ \one_S\|_1} \right\}$}
      \State Let $\eta_2$ be the maximum $\eta \in [0, \eps^2]$ such that $\|\zet\|_1 \leq \eps jk$ \label{line:eta2}
      \State $\eta \gets \min\{\eta_1, \eta_2\}$ \label{line:eta}
      \State $\x \gets \x + \eta (\one - \x) \circ \one_{\Tet}$ \label{line:update-x}
      \State $\z \gets \z + \eta (\one - \z) \circ \one_S$ \label{line:update-z}
      \If{$f(\z) > f(\x)$}
        \State $\x \gets \z$ \label{line:fz-larger}
      \EndIf
    \EndIf
  \EndWhile
\EndFor
\State \Return $\x$
\end{algorithmic}
\end{algorithm}

{\bf Finding the step size $\eta_1$ on line~\ref{line:eta1}.} As mentioned earlier, we assume that we can find the step $\eta_1$ exactly. In the appendix, we show that we can efficiently find $\eta_1$ approximately using $t$-ary search for suitable $t$. We can choose $t$ to obtain different trade-offs between the number of parallel rounds and total running time, see Section~\ref{app:approx-step-size} in the appendix for more details.

\medskip
{\bf Finding the step size $\eta_2$ on line~\ref{line:eta2}.} We have $\vec{z}(\eta) \geq \vec{0}$ and
\begin{align*}
\sum_{i \in [n]} \vec{z}_i(\eta) &= \sum_{i \in S} (\vec{z}_i + \eta (1 - \vec{z}_i)) + \sum_{i \notin S} \vec{z}_i\\
&= \sum_{i \in [n]} \vec{z}_i + \eta \sum_{i \in S}(1 - \vec{z}_i)
\end{align*}
Additionally, for each $i \in S$, we have $\vec{z}_i \leq 1 - e^{-t} < 1$ and thus $\sum_{i \in S} (1 - \vec{z}_i) > 0$. Therefore $\eta_2$ is the minimum between $\eps^2$ and the following value:
\[ \frac{\eps jk - \sum_{i \in [n]} \vec{z}_i}{\sum_{i \in S}(1 - \vec{z}_i)} = \frac{\eps jk - \|\vec{z}\|_1}{|S| - \|\vec{z} \circ \vec{1}_S\|_1} \]

\section{Analysis of the approximation guarantee}
\label{sec:approx}

In this section, we show that Algorithm~\ref{alg:non-monotone-ideal} achieves a $\frac{1}{e} - O(\eps)$ approximation. Recall that we assume that $\eta_1$ is computed exactly on line~\ref{line:eta1}. In Section~\ref{app:approx-step-size} of the appendix, we show how to extend the algorithm and the analysis so that the algorithm efficiently computes a suitable approximation to $\eta_1$ that suffices for obtaining a $\frac{1}{e} - O(\eps)$ approximation.

In the following, we refer to each iteration of the outer for loop as a \emph{phase}. We refer to each iteration of the inner while loop as an iteration. Note that the update vectors are non-negative in each iteration of the algorithm, and thus the vectors $\x, \z$ remain non-negative throughout the algorithm and they can only increase. Additionally, since $\Set \subseteq \Tet \subseteq S$, we have $\x \leq \z$ throughout the algorithm. We will also use the following observations repeatedly, whose straightforward proofs are deferred to Section~\ref{app:omitted} of the appendix. By DR-submodularity, since the relevant vectors can only increase in each coordinate, the relevant gradients can only decrease in each coordinate. This implies that, for every $\eta \leq \eta'$, we have $S(\eta) \supseteq S(\eta')$. Additionally, for every $i \in \Tet$, we have $\nabla_i f(\x) \geq \nabla_i f(\z) \geq \nabla_i f(\zet) > 0$.

We will need an upper bound on the $\ell_1$ and $\ell_{\infty}$ norms of $\x$ and $\z$. Since $\x \leq \z$, it suffices to upper bound the norms of $\z$ (the $\ell_1$ norm bound will be used to show that the final solution is feasible, and the $\ell_{\infty}$ norm bound will be used to derive the approximation guarantee). We do so in the following lemma.

\begin{lemma}
\label{lem:z-norms}
Consider phase $j$ of the algorithm (the $j$-th iteration of the outer for loop). Throughout the phase, the algorithm maintains the invariant that $\|\vec{z}\|_{\infty} \leq 1 - (1 - \eps)^j + \eps^2$ and $\|\vec{z}\|_1 \leq \eps j k$.
\end{lemma}
\begin{proof}
We show that the invariants are maintained using induction on the number of iterations of the inner while loop in phase $j$. Let $\z$ be the vector right before the update on line~\ref{line:update-z} and let $\zp$ be the vector right after the update. By the induction hypothesis, we have $\z_i \leq 1 - (1 - \eps)^j + \eps^2$. If $i \notin S$, we have $\vec{z}'_i = \vec{z}_i$, and the invariant is maintained. Therefore we may assume that $i \in S$. By the definition of $S$, we have $\z_i \leq 1 - (1 - \eps)^j$. We have $\zp_i = \z_i + \eta (1 - \zp_i) \leq \z_i + \eta \leq \z_i + \eps^2$. Thus the invariant is maintained.

Next, we show the upper bound on the $\ell_1$ norm. Note that $\zp = \zet \leq \z(\eta_2)$, where $\eta$ is the step size chosen on line~\ref{line:eta}. Thus we have $\|\zp\|_1 \leq \|\vec{z}(\eta_2)\|_1 \leq \eps jk$, where the last inequality is by the choice of $\eta_2$.
\end{proof}

\begin{theorem}
\label{thm:phase-gain}
Consider a phase of the algorithm (an iteration of the outer for loop). Let $\xstart$ and $\xend$ be the vector $\x$ at the beginning and end of the phase. We have
\[ f(\xend) - f(\xstart) \geq (1 - 5\eps) \eps ((1 - \eps)^j f(\xopt) - f(\xend) - 3 \eps f(\xopt)) \]
\end{theorem}
\begin{proof}
We consider two cases, depending on whether the threshold $\vend$ at the end of the phase is equal to $\vstart$ or not.

\medskip\noindent
{\bf Case 1:} we have $\vend = \vstart$. Note that the phase terminates with $\|\zend\|_1 = \eps jk$ in this case. We fix an iteration of the phase that updates $\x$ and $\z$ on lines~\ref{line:update-x}--\ref{line:fz-larger}, and analyze the gain in function value in the current iteration. We let $\x, \z$ denote the vectors right before the update on lines~\ref{line:update-x}--\ref{line:fz-larger}. Let $\xp$ be the vector $\x$ right after the update on line~\ref{line:update-x}, and let $\zp$ be the vector $\z$ right after the update on line~\ref{line:update-z}.

We have:
\begin{align*}
f(\xp) - f(\x)
&\overset{(a)}{\geq} \langle \nabla f(\xp), \eta (\one - \x) \circ \one_{\Tet} \rangle\\
&= \langle (\one - \x) \circ \nabla f(\xp), \eta \one_{\Tet} \rangle\\
&\overset{(b)}{\geq} \langle \get, \eta \one_{\Tet} \rangle\\
&\overset{(c)}{\geq}  \eta \vstart \underbrace{|\Set|}_{\geq (1 - \eps) |S|}\\
&\overset{(d)}{\geq} (1 - \eps) \eta \vstart |S|
\end{align*}
In (a), we used the fact that $\vec{x}' - \vec{x} \geq 0$ and $f$ is concave in non-negative directions.

We can show (b) as follows. We have $\xp \leq \zp = \zet$ and thus $\nabla f(\xp) \geq \nabla f(\zet)$ by DR-submodularity. Additionally, for every coordinate $i \in \Tet$, we have $\nabla_i f(\zet) > 0$. Therefore, for every $i \in \Tet$, we have $(1 - \x_i) \nabla_i f(\xp) \geq (1 - \zet_i) \nabla_i f(\zet) = \get_i > 0$.

In (c), we have used that $\Set \subseteq \Tet$, $\get_i > 0$ for all $i \in \Tet$, and $\get_i \geq v = \vstart$ for all $i \in \Set$.

We can show (d) as follows. Since $\eta \leq \eta_1$, we have $|S(\eta)| \geq |S(\eta_1)| \geq (1 - \eps) |S|$, where the first inequality is by Lemma~\ref{lem:monotonicity} and the second inequality is by the choice of $\eta_1$.

Let $\eta_t$ and $S_t$ denote $\eta$ and $S$ in iteration $t$ of the phase (note that we are momentarily overloading $\eta_1$ and $\eta_2$ here, and they temporarily stand for the step size $\eta$ in iterations $1$ and $2$, and not for the step sizes on lines~\ref{line:eta1} and \ref{line:eta2}). By summing up the above inequality over all iterations, we obtain:
\begin{align*}
f(\xend) - f(\xstart)
&\geq (1 - \eps) \vstart \sum_t \eta_t |S_t|\\
&\geq  (1 - \eps) \vstart \underbrace{\|\zend - \zstart\|_1}_{\geq \eps k}\\
&\overset{(a)}{\geq} (1 - \eps) \vstart \eps k\\
&\overset{(b)}{=} \eps (1 - \eps) (((1 - \eps)^j - 2\eps) M - f(\xstart))\\
&\overset{(c)}{\geq} \eps (1 - \eps) ((1 - \eps)^j f(\xopt) - f(\xstart) - 3\eps M)
\end{align*}
We can show (a) as follows. Recall that we have $\|\zend\|_1 = \eps j k$. Since $\|\zstart \|_1 \leq \eps (j - 1) k$, we have $\|\zend - \zstart\|_1 \geq \eps k$.

In (b), we used the definition of $\vstart$ on line~\ref{line:vstart}.

In (c), we used that $f(\xopt) \leq (1 + \eps) M$.

\medskip\noindent
{\bf Case 2:} we have $\vend \neq \vstart$. Note that this implies that $\vend \leq (1 - \eps) \vstart$, since line~\ref{line:update-v} was executed at least once during the phase.

Let $A$ be the following subset of the coordinates:
\[ A := \left\{i \in [n] \colon (1 - (\zend)_i) \nabla_i f(\zend) \geq \frac{\vend}{1 - \eps} \right\} \]

\begin{lemma}
\label{lem:phase-gain-common}
We have
\[f(\xend) - f(\xstart) \geq (1 - \eps) \Big( \vend \|(\zend - \zstart) \circ \one_{\overline{A}}\|_1 + \langle \nabla f(\zend), (\zend - \zstart) \circ \one_{A} \rangle \Big)
 \]
\end{lemma}
\begin{proof}
Fix an iteration of the phase that updates $\x$ and $\z$ on lines~\ref{line:update-x}--\ref{line:fz-larger}. Let $\x, \z$ denote the vectors right before the update on lines~\ref{line:update-x}--\ref{line:fz-larger}. Let $\xp$ be the vector $\x$ right after the update on line~\ref{line:update-x}, and let $\zp$ be the vector $\z$ right after the update on line~\ref{line:update-z}.

We have:
\begin{align*}
f(\xp) - f(\x)
&\overset{(a)}{\geq} \langle \nabla f(\xp), \eta (\one - \x) \circ \one_{\Tet} \rangle\\
&= \langle (\one - \x) \circ \nabla f(\xp), \eta \one_{\Tet} \rangle\\
&\overset{(b)}{\geq} \langle \get, \eta \one_{\Tet} \rangle\\
&= \langle \get, \eta \one_{\Set} \rangle +  \langle \get, \eta \one_{\Tet \setminus \Set} \rangle\\
&= \eta \vend \underbrace{|\Set|}_{\geq (1 - \eps) |S|} + \underbrace{\langle \get - \vend \one, \eta \one_{\Set} \rangle }_{\geq 0} + \underbrace{\langle \get, \eta \one_{\Tet \setminus \Set} \rangle}_{\geq 0} \\
&\overset{(c)}{\geq} (1 - \eps) \eta \left( \vend |S| + \langle \get - \vend \one, \one_{\Set} \rangle + \langle \get, \one_{\Tet \setminus \Set} \rangle \right)\\
&= (1 - \eps) \eta \left( \vend |S| - \vend |\Set| + \langle \get, \one_{\Tet} \rangle \right)\\
&= (1 - \eps) \eta \left( \vend |S| - \vend |\Set| + \langle \get, \one_{\Tet \setminus A} \rangle + \langle \get, \one_{\Tet \cap A} \rangle \right)\\
&\overset{(d)}{=} (1 - \eps) \eta \left( \vend |S| - \vend |\Set| + \langle \get, \one_{\Tet \setminus A} \rangle + \langle \get, \one_{S \cap A} \rangle \right)\\
&\overset{(e)}{\geq} (1 - \eps) \eta \left( \vend |S| - \vend |\Set| + \langle \get, \one_{\Set \setminus A} \rangle + \langle \get, \one_{S \cap A} \rangle \right)\\
&\overset{(f)}{\geq} (1 - \eps) \eta \left( \vend |S| - \vend |\Set| + \vend |\Set \setminus A| + \langle \get, \one_{S \cap A} \rangle \right)\\
&= (1 - \eps) \eta \left( \vend (|S| - |\Set \cap A|) + \langle \get, \one_{S \cap A} \rangle \right)\\
&\geq (1 - \eps) \eta \left( \vend |S \setminus A| + \langle \get, \one_{S \cap A} \rangle \right)\\ 
&= (1 - \eps) \eta \left( \vend |S \setminus A| + \langle \nabla f(\zet), (\one - \zet) \circ \one_{S \cap A} \rangle \right)\\
&\overset{(g)}{\geq} (1 - \eps) \eta \left( \vend |S \setminus A| + \langle \nabla f(\zend), (\one - \zet) \circ \one_{S \cap A} \rangle \right)
\end{align*}
In (a), we used the fact that $\vec{x}' - \vec{x} \geq 0$ and $f$ is concave in non-negative directions.

We can show (b) as follows. We have $\xp \leq \zp = \zet$ and thus $\nabla f(\xp) \geq \nabla f(\zet)$ by DR-submodularity. Additionally, for every coordinate $i \in \Tet$, we have $\nabla_i f(\zet) > 0$. Therefore, for every $i \in \Tet$, we have $(1 - \x_i) \nabla_i f(\xp) \geq (1 - \zet_i) \nabla_i f(\zet) = \get_i > 0$.

We can show (c) as follows. Since $\eta \leq \eta_1$, we have $|S(\eta)| \geq |S(\eta_1)| \geq (1 - \eps) |S|$, where the first inequality is by Lemma~\ref{lem:monotonicity} and the second inequality is by the choice of $\eta_1$. By the definition of $\Set$, we have $\get_i \geq v \geq \vend$ for every $i \in \Set$. By the definition of $\Tet$, we have $\get_i > 0$ for every $i \in \Tet$.

Equality (d) follows from the fact that $S \cap A = T(\eta) \cap A$, which we can show as follows. We have $T(\eta) \subseteq S$, and $S \setminus T(\eta)$ is the set of all coordinates with negative gradient $\get$. Thus it suffices to show that the coordinates in $A$ have positive gradient $\get$. For every $i \in A$, we have $\nabla_i f(\zet) \geq \nabla_i f(\zend) > 0$, where the first inequality is by DR-submodularity (since $\zet \leq \zend$) and the second inequality is by the definition of $A$ and the fact that $(\zend)_i < 1$ for all $i \in [n]$ (Lemma~\ref{lem:z-norms}). Moreover, we have $\zet_i < 1$ for all $i \in [n]$ (Lemma~\ref{lem:z-norms}). Thus $\get_i > 0$ for all $i \in A$, and hence $S \cap A = T(\eta) \cap A$.

In (e), we have used that $\Set \subseteq \Tet$ and $\get$ is non-negative on the coordinates of $\Tet$.

In (f), we have used that, $\get_i \geq v \geq \vend$ for all $i \in \Set$ .

In (g), we have used that $\zet \leq \zend$ and thus $\nabla f(\zet) \geq \nabla f(\zend)$ by DR-submodularity.

Let $\eta_t, S_t, T_t(\eta), \z_t(\eta), \g_t(\eta)$ denote $\eta, S, T(\eta), \zet, \get$ in iteration $t$ of the phase (note that we are momentarily overloading $\eta_1$ and $\eta_2$, and they temporarily stand for the step size $\eta$ in iterations $1$ and $2$, and not for the step sizes on lines~\ref{line:eta1} and \ref{line:eta2}). By summing up the above inequality over all iterations, we obtain:
\begin{align*}
f(\xend) - f(\xstart)
&\geq (1 - \eps) \Bigg( \sum_t \vend \eta_t |S_t \setminus A| + \sum_t \langle \nabla f(\zend), \eta_t (\one - \zet) \circ \one_{S_t \cap A} \rangle \Bigg)\\
&\geq (1 - \eps) \Bigg( \vend \|(\zend - \zstart) \circ \one_{\overline{A}}\|_1 + \langle \nabla f(\zend), (\zend - \zstart) \circ \one_{A} \rangle \Bigg)
\end{align*}
\end{proof}

We will also need the following lemmas.

\begin{lemma}
\label{lem:delta-z-A}
For every $i \in A$, we have:
\[ (\zend)_i - (\zstart)_i \geq (1 - 3\eps) \eps (1 - (\zend)_i) \] 
\end{lemma}
\begin{proof}
Since $S$ was empty at the previous threshold $\vend / (1 - \eps)$, we have $(\zend)_i \geq 1 - (1 - \eps)^j$ or $(\zend)_i - (\zstart)_i \geq \eps (1 - (\zstart)_i)$. If it is the latter, the claim follows, since $1 - (\zstart)_i \geq 1 - (\zend)_i$. Therefore we may assume it is the former. By Lemma~\ref{lem:z-norms}, $(\zstart)_i \leq 1 - (1 - \eps)^{j - 1} + \eps^2$. Therefore 
\[ (\zend)_i - (\zstart)_i \geq \eps (1 - \eps)^{j - 1} - \eps^2 \geq (1 - 3\eps) \eps (1 - \eps)^{j - 1} \geq (1 - 3\eps) \eps (1 - (\zend)_i) \]
where in the second inequality we used that $(1 - \eps)^{j - 1} \geq 1/3$ for sufficiently small $\eps$ (since $(1 - \eps)^{j - 1} \geq (1 - \eps)^{1/\eps} \approx 1/e$).
\end{proof}

\begin{lemma}
\label{lem:grad-z-xopt}
We have:
\[ \langle \nabla f(\zend) \vee \zero, (\one - \zend) \circ \xopt \rangle \geq ((1 - \eps)^j - \eps^2) f(\xopt) - f(\xend) \]
\end{lemma}
\begin{proof}
We have:
\begin{align*}
\langle \nabla f(\zend) \vee \zero, (\one - \zend) \circ \xopt \rangle
&\overset{(a)}{\geq} \langle \nabla f(\zend) \vee \zero, \xopt \vee \zend - \zend \rangle\\
&\overset{(b)}{\geq} f(\zend \vee \xopt) - f(\zend)\\
&\overset{(c)}{\geq} f(\zend \vee \xopt) - f(\xend)\\
&\overset{(d)}{\geq} (1 - \|\zend\|_{\infty}) f(\xopt) - f(\xend)\\
&\overset{(e)}{\geq} ((1 - \eps)^j - \eps^2) f(\xopt) - f(\xend)
\end{align*}
In (a), we used that $(1 - a) b \geq \max\{a, b\} - a$ for all $a, b \in [0, 1]$.

In (b), we used the fact that $f$ is concave in non-negative directions.

In (c), we used the fact that the algorithm maintains the invariant that $f(\x) \geq f(\z)$ via the update on line~\ref{line:fz-larger}.

In (d), we used Lemma~\ref{lem:x-or-opt}.

In (e), we used Lemma~\ref{lem:z-norms}.
\end{proof}

Recall that the phase terminates with either $\vend \leq \eps \vstart$ or $\|\zend\|_1 = \eps jk$. We consider each of these cases in turn.

\begin{lemma}
\label{lem:phase-gain-small-threshold}
Suppose that $\vend \leq \eps \vstart$. We have:
\[ f(\xend) - f(\xstart) \geq (1 - 5\eps) \eps ((1 - \eps)^j f(\xopt) - f(\xend) - 2\eps f(\xopt)) \]
\end{lemma}
\begin{proof}
By Lemma~\ref{lem:phase-gain-common}, we have:
\begin{align*}
&f(\xend) - f(\xstart)\\
&\geq (1 - \eps) \langle \nabla f(\zend), (\zend - \zstart) \circ \one_{A} \rangle \\
&= (1 - \eps) \Big( \underbrace{\langle \nabla f(\zend), (\zend - \zstart - \eps (1 - 3\eps) (\one - \zend) \circ \xopt) \circ \one_{A} \rangle}_{\geq 0}\\
&\quad\qquad\qquad + \langle \nabla f(\zend), \eps (1 - 3\eps) (\one - \zend) \circ \xopt \circ \one_{A} \rangle \Big)\\
&\overset{(a)}{\geq} (1 - 4\eps) \eps  \langle \nabla f(\zend), (\one - \zend) \circ \xopt \circ \one_{A} \rangle\\
&= (1 - 4\eps) \eps  \langle (\one - \zend) \circ \nabla f(\zend), \xopt \circ \one_{A} \rangle\\
&= (1 - 4\eps) \eps \Big( \langle (\one - \zend) \circ \nabla f(\zend) \vee \zero, \xopt \rangle - \langle (\one - \zend) \circ \nabla f(\zend) \vee \zero, \xopt \circ \one_{\overline{A}} \rangle  \Big)\\
&\overset{(b)}{\geq} (1 - 4\eps) \eps \Big( \langle (\one - \zend) \circ \nabla f(\zend) \vee \zero, \xopt \rangle - \frac{\eps \vstart}{1 - \eps} k  \Big)\\
&\geq (1 - 5\eps) \eps \Big( \langle (\one - \zend) \circ \nabla f(\zend) \vee \zero, \xopt \rangle - \eps \vstart k  \Big)\\
&\overset{(c)}{\geq} (1 - 5\eps) \eps ( ((1 - \eps)^j - \eps^2) f(\xopt) - f(\xend) - \eps \vstart k)\\
&\geq (1 - 5\eps) \eps ( ((1 - \eps)^j - \eps^2) f(\xopt) - f(\xend) - \eps f(\xopt))\\
&\geq (1 - 5\eps) \eps ((1 - \eps)^j f(\xopt) - f(\xend) - 2\eps f(\xopt))
\end{align*}
In (a), we used Lemma~\ref{lem:delta-z-A} and the fact that $\zero \leq \xopt \leq \one$ and the fact that $\nabla_i f(\zend) > 0$ for all $i \in A$. 

In (b), we used that $\|\xopt\|_1 \leq k$ and, for all $i \notin A$:
\[ (1 - (\zend)_i) \nabla_i f(\zend) \leq \frac{\vend}{1 - \eps} \leq \frac{\eps \vstart}{1 - \eps} \] 

In (c), we used Lemma~\ref{lem:grad-z-xopt}.

Inequality (d) follows from the definition of $\vstart$ and the fact that $f(\xopt) \geq M$.
\end{proof}

\begin{lemma}
\label{lem:phase-gain-constraint}
Suppose that $\|\zend\|_1 = \eps jk$. We have:
\[ f(\xend) - f(\xstart) \geq \eps (1 - 4\eps)  (((1 - \eps)^j - \eps^2) f(\xopt) - f(\xend)) \]
\end{lemma}
\begin{proof}
By Lemma~\ref{lem:phase-gain-common}, we have:
\begin{align*}
&f(\xend) - f(\xstart)\\
&\geq (1 - \eps) \Bigg( \vend \|(\zend - \zstart) \circ \one_{\overline{A}}\|_1 + \langle \nabla f(\zend), (\zend - \zstart) \circ \one_{A} \rangle \Bigg)\\
&= (1 - \eps) \Bigg( \vend \|(\zend - \zstart) \circ \one_{\overline{A}}\|_1 + \langle \nabla f(\zend), (\zend - \zstart - (1 - 3\eps) \eps (\one - \zend) \circ \xopt) \circ \one_{A} \rangle\\
&\quad\qquad\qquad + \langle \nabla f(\zend), (1 - 3\eps) \eps (\one - \zend) \circ \xopt \circ \one_{A} \rangle \Bigg)\\
&\overset{(a)}{\geq} (1 - \eps) \Bigg( \vend \|(\zend - \zstart) \circ \one_{\overline{A}}\|_1 + \vend \|(\zend - \zstart - (1 - 3\eps)\eps (\one - \zend) \circ \xopt) \circ \one_{A}\|_1\\
&\quad\qquad\qquad + (1 - 3\eps) \eps \langle \nabla f(\zend), (\one - \zend) \circ \xopt \circ \one_{A} \rangle \Bigg)\\
&= (1 - \eps) \Bigg( \vend \underbrace{\|\zend - \zstart\|_1}_{\geq \eps k} - \vend \underbrace{\|(1 - 3\eps) \eps (\one - \zend) \circ \xopt \circ \one_{A}\|_1}_{\leq \eps \|\xopt \circ \one_A\|_1}\\
&\quad\qquad\qquad + (1 - 3\eps) \eps \langle \nabla f(\zend), (\one - \zend) \circ \xopt \circ \one_{A} \rangle  \Bigg)\\
&\overset{(b)}{\geq} (1 - \eps) \Bigg( \vend \eps \|\xopt \circ \one_{\overline{A}}\|_1 + (1 - 3\eps) \eps \langle \nabla f(\zend), (\one - \zend) \circ \xopt \circ \one_{A} \rangle  \Bigg)\\
&\overset{(c)}{\geq} (1 - \eps) \Bigg( (1 - \eps) \eps \langle \nabla f(\zend) \vee \zero, (\one - \zend) \circ \xopt \circ \one_{\overline{A}} \rangle + (1 - 3\eps) \eps \langle \nabla f(\zend), (\one - \zend) \circ \xopt \circ \one_{A} \rangle \Bigg)\\
&\geq (1 - 4\eps) \eps \langle \nabla f(\zend) \vee \zero, (\one - \zend) \circ \xopt \rangle\\
&\overset{(d)}{\geq} (1 - 4\eps) \eps  (((1 - \eps)^j - \eps^2) f(\xopt) - f(\xend))
\end{align*}
In (a), we used Lemma~\ref{lem:delta-z-A} and the fact that $\zero \leq \xopt \leq \one$ and the fact that $\nabla_i f(\zend) \geq \vend$ for all $i \in A$. 

We can show (b) as follows. Recall that we are in the case $\|\zend\|_1 = \eps j k$. Since $\|\zstart \|_1 \leq \eps (j - 1) k$, we have $\|\zend - \zstart\|_1 \geq \eps k$. Additionally, $\|\xopt\|_1 \leq k$.

In (c), we used that, for all $i \notin A$, we have $\vend \geq (1 - \eps) (1 - (\zend)_i) \nabla_i f(\zend)$.

In (d), we used Lemma~\ref{lem:grad-z-xopt}.
\end{proof}
\end{proof}

Using induction and Theorem~\ref{thm:phase-gain}, we can show that the final solution returned by the algorithm is a $1/e - O(\eps)$ approximation. By construction, the final solution satisfies $\|\x\|_1 \leq k$, and thus it also satisfies the constraint.

\begin{theorem}
\label{thm:approx}
Let $\x$ be the final solution returned by Algorithm~\ref{alg:non-monotone}. We have $\|\x\|_1 \leq k$ and $f(\x) \geq \left(\frac{1}{e} - O(\eps) \right) f(\xopt)$.
\end{theorem}
\begin{proof}
By Lemma~\ref{lem:z-norms}, the algorithm maintains the invariant that, at the end of phase $j$, we have $\|\x\|_1 \leq \|\z\|_1 \leq \eps j k$. Thus, at the end of the algorithm, we have $\|\x\|_1 \leq k$.

Next, we show the approximation guarantee. Let $\x^{(0)} = \zero$ and let $\x^{(j)}$ be the solution $\x$ at the end of phase $j$. We will show by induction on $j$ that:
\[ f(\x^{(j)}) \geq \eps j (1 - \eps)^j f(\xopt) - 9j \eps^2 f(\xopt) \]
The above inequality clearly hods for $j = 0$. Consider $j \geq 1$. By Theorem~\ref{thm:phase-gain}, we have
\begin{align*}
f(\x^{(j)}) &\geq f(\x^{(j - 1)}) + (1 - 5\eps) \eps ((1 - \eps)^j f(\xopt) - f(\x^{(j)}) - 3 \eps f(\xopt))\\
\Rightarrow (1 + \eps) f(\x^{(j)}) &\geq f(\x^{(j - 1)}) + (1 - 5\eps) \eps ((1 - \eps)^j f(\xopt) - 3 \eps f(\xopt))\\
\Rightarrow f(\x^{(j)}) &\geq (1 - \eps) f(\x^{(j - 1)}) + (1 - 6\eps) \eps ((1 - \eps)^j f(\xopt) - 3 \eps f(\xopt))\\
&\geq (1 - \eps) f(\x^{(j - 1)}) + \eps (1 - \eps)^j f(\xopt) - 9 \eps^2 f(\xopt)\\
&\overset{(a)}{\geq} (1 - \eps) \Big(\eps (j - 1) (1 - \eps)^{j - 1} f(\xopt) - 9 (j - 1) \eps^2 f(\xopt) \Big) + \eps (1 - \eps)^j f(\xopt) - 9 \eps^2 f(\xopt)\\
&\geq \eps j (1 - \eps)^j f(\xopt) - 9j \eps^2 f(\xopt)
\end{align*}
where (a) is by the inductive hypothesis.

Thus it follows by induction that:
\[ f(\x^{(1/\eps)}) \geq ((1 - \eps)^{1/\eps} - 9\eps) f(\xopt) \geq \left(\frac{1}{e} - O(\eps)\right) f(\xopt),\]
as needed.
\end{proof}

\section{Analysis of the number of iterations}
\label{sec:iters}

Recall that we refer to each iteration of the outer for loop as a \emph{phase}. We refer to each iteration of the inner while loop as an iteration. 

\begin{theorem}
\label{thm:iters}
The total number of iterations of the algorithm is $O(\log(n)\log(1/\eps)/\eps^3)$.
\end{theorem}
\begin{proof}
There are $O(1/\eps)$ phases. In each phase, there are $O(\log(1/\eps)/\eps)$ different thresholds $v$: the initial threshold is $\vstart$, the threshold right before the final one is at least $\eps \vstart$, and each update on line~\ref{line:update-v} decreases the threshold by a $(1 - \eps)$ factor. Thus it only remains to bound the number of iterations with the same threshold.

In the following, we fix a single threshold and we consider only the iterations of the phase at that threshold. Over these iterations, $\z$ is non-decreasing in every coordinate, $\g$ is non-increasing in every coordinate by DR-submodularity and $\one - \z \geq \zero$, and the set $S$ can only lose coordinates and thus $|S|$ is non-increasing. Additionally, for each coordinate $i \in [n]$, the increase $(\zend)_i - (\zstart)_i$ over the entire phase is at most $\eps + \eps^2$: the increase in each iteration is $\eta$ if $i \in S$ and $0$ otherwise; since $\eta \leq \eps^2$ and $\z_i - (\zstart)_i < \eps (1 - (\zstart)_i)$ for every $i \in S$, the claim follows.  

We say that an iteration is a \emph{large-step iteration} if $\eta = \eps^2$ and it is a \emph{smaller-step iteration} if $\eta < \eps^2$. 

We first consider the large-step iterations. Let $t$ be the last large-step iteration, and let $S_t$ be the set $S$ in that iteration. Let $i \in S_t$. Note that $i \in S_{t'}$ for all iterations $t' \leq t$, since $S_{t} \subseteq S_{t'}$. Thus every large-step iteration increases $\z_i$ by $\eps^2 (1 - \z_i) \geq \eps^2 (1 - \eps)^j \geq \eps^2 (1 - \eps)^{1/\eps} = \Theta(\eps^2)$. Since $\z_i$ increases by at most $\eps + \eps^2$ over the entire phase, it follows that the number of large-step iterations is $O(1/\eps)$. 

Next, we consider the smaller-step iterations. Note that, in every smaller-step iteration except possibly the last one, we have $|S(\eta)| \leq (1 - \eps) |S|$ (if $\eta = \eta_2 < \eps^2$, at the end of the iteration we have $\|\z\|_1 = \eps jk$ and thus the phase ends; if $\eta = \eta_1 < \eps^2$, our choice of $\eta_1$ ensures that $|S(\eta_1)| \leq (1 - \eps) |S|$). Thus every smaller-step iteration decreases $|S|$ by at least an $(1 - \eps)$ factor. Now note that $|S| \leq n$ in the first iteration, $|S| \geq 1$ in the last iteration, and $|S|$ can only decrease with each iteration. Thus the number of smaller-step iterations is $O(\log{n}/\eps)$.

In summary, there are $O(1/\eps)$ phases, $O(\log(1/\eps)/\eps)$ different thresholds per phase, and $O(\log{n}/\eps)$ iterations per threshold. Thus the total number of iterations of the algorithm is $O(\log{n} \log(1/\eps)/\eps^3)$. 
\end{proof}

\section{Experimental Results}

\begin{figure*}[t]
\subfigure[][NQP instances: function value]{
\begin{minipage}[t]{0.49\textwidth}
	\centering
	\includegraphics[width=\textwidth]{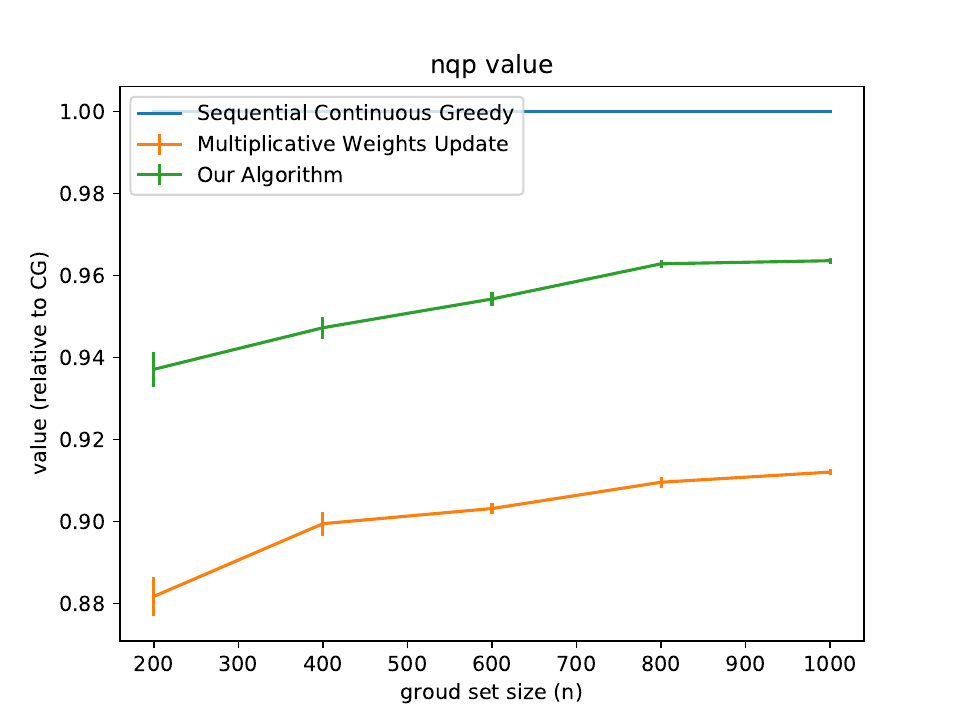}
\end{minipage}
}
\subfigure[][NQP instances: adaptive evaluations]{
\begin{minipage}[t]{0.49\textwidth}
	\centering
	\includegraphics[width=\textwidth]{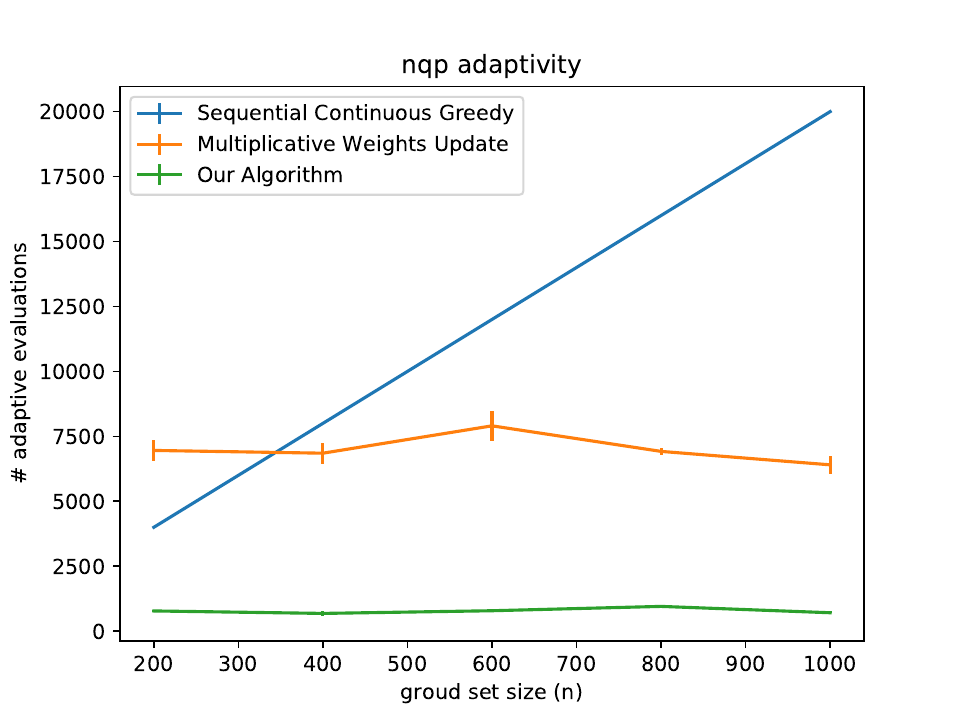}
\end{minipage}
}
\subfigure[][DPP instances: function value]{
\begin{minipage}[t]{0.49\textwidth}
	\centering
	\includegraphics[width=\textwidth]{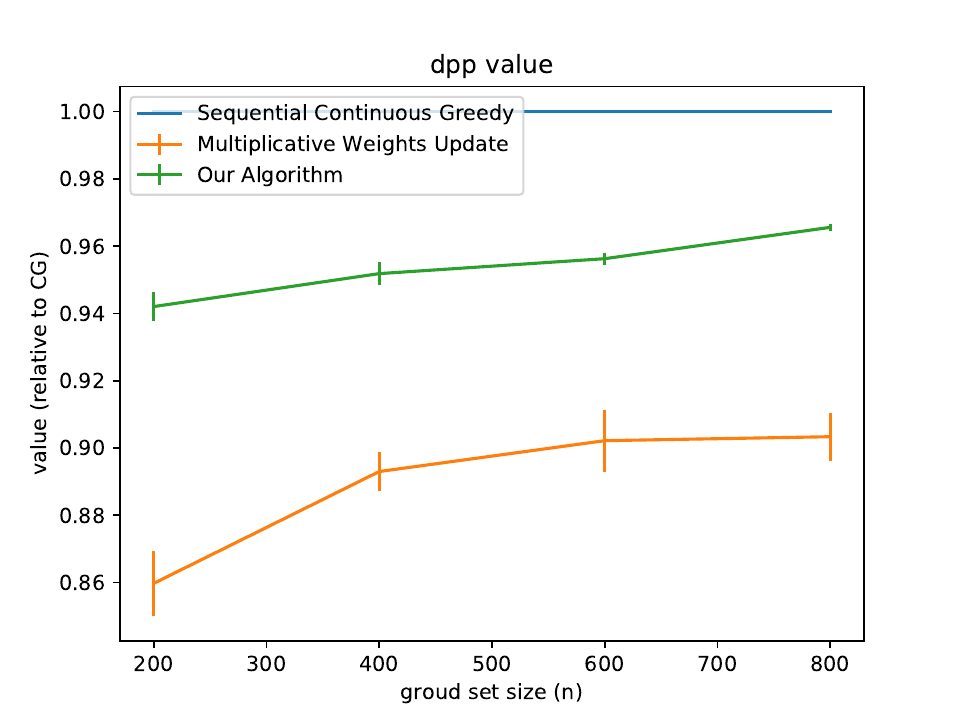}
\end{minipage}
}
\subfigure[][DPP instances: adaptive evaluations]{
\begin{minipage}[t]{0.49\textwidth}
	\centering
	\includegraphics[width=\textwidth]{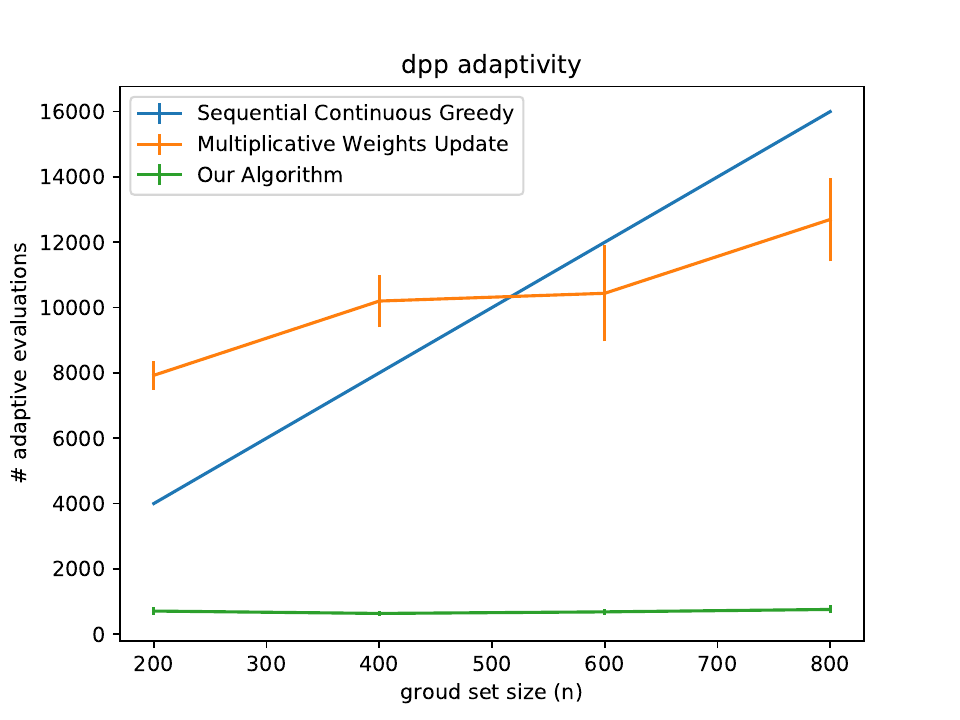}
\end{minipage}
}
\caption{Experimental results.}
\label{fig:experiments}
\end{figure*}

We experimentally evaluate our parallel algorithm on instances of non-concave quadratic programming (NQP) and softmax extension of determinantal point processes (DPP). We randomly generated NQP and DPP functions using the following approach that is similar to previous work~\cite{bian2017continuous}.

\emph{NQP instances} are functions of the form $f(\x) = \frac{1}{2} \x^\top H \x + \vec{h}^\top \x$, where $H \in \mathrm{R}^{n \times n}$ is a matrix with non-positive entries, $\vec{h} \in \R^n$. We randomly generated such instances as follows: we sampled each entry of $H$ uniformly at random from $[-10, 0]$, and we set $\vec{h} = - 0.2 H^\top \vec{1}$.

\emph{DPP instances} are functions of the form $f(\x) = \log\det(\mathrm{diag}(\x) (L - I) + I)$, where $L \in \R^{n \times n}$ is a psd matrix and $I$ is the identity matrix. We randomly generated such instances as follows. We sampled the eigenvalues of $L$ as follows: the $i$-th eigenvalue is $\ell_i = e^{r_i}$, where $r_i$ was sampled uniformly from $[-0.5, 1]$. We sampled a random orthogonal matrix $V$. We set $L = V \mathrm{diag}(\ell_1, \dots, \ell_n) V^\top$. 

\emph{Algorithms, implementation details, and parameter choices.}
We empirically compared our parallel algorithm with the state of the art sequential and parallel algorithms, which we describe in more detail in Section~\ref{app:algorithms} of the appendix. The sequential algorithm that we used in our experiments is a variant of the measured continuous greedy algorithm that was studied in the works~\cite{ChekuriJV15,bian2017continuous}; this algorithm outperforms the standard measured continuous greedy algorithm in terms of solution quality since it fills up more of the available budget, while performing the same number of iterations. We implemented the sequential continuous greedy algorithm using a step size of $\eps / n$, leading to $O(n/\eps)$ iterations and adaptive evaluations\footnote{The theoretical guarantee of $1/e - \eps$ for the measured continuous greedy and related algorithms is obtained with the more conservative step size of $\eps / n^3$, and thus $O(n^3/\eps)$ iterations, but such a high number of iterations was prohibitive in our experiments.}. The state of the art parallel algorithm for non-mononotone DR-submodular maximization with a cardinality constraint is the multiplicative weights update algorithm of \cite{ENV19} that achieves a $1/e - \eps$ approximation using $O(\log^2{n}/\eps^2)$ iterations. We implemented our algorithm with a more aggressive update of the thresholds on line~\ref{line:update-v}: instead of the update $v \gets (1 - \eps) v$, we performed the update $v \gets 0.75 \cdot v$, i.e., the threshold is updated by a constant factor independent of $\eps$ instead of $1-\eps$. Thus the algorithm only performs $O(\log{n} \log(1/\eps)/ \eps^2)$ iterations. We used error $\eps = 0.05$ and budget $k = 10$ in all of the experiments.

\emph{Computing infrastructure.} We implemented the algorithms in C++ and ran the experiments on an iMac with a 3.3 GHz Intel Core i5 processor and 8 GB of memory.

\emph{Results.} The experimental results are shown in Figure~\ref{fig:experiments}. Each value is the average value for $5$ independently sampled instances and the error bar is $\pm 1$ standard deviation. The sequential continuous greedy algorithm achieved the highest solution value in all of the runs, and we report the value obtained by the parallel algorithms as the fraction of the continuous greedy solution value. In all of the runs, our parallel algorithm achieves higher function value than the parallel multiplicative weights update algorithm, while the number of evaluations is significantly lower.

\clearpage

\bibliographystyle{abbrv}
\bibliography{../submodular}

\newpage
\appendix

\section{Omitted proofs}
\label{app:omitted}

\begin{lemma}
\label{lem:xz}
The algorithm maintains the invariant that $\vec{0} \leq \vec{x} \leq \vec{z}$. 
\end{lemma}
\begin{proof}
We show the lemma by induction on the number of updates (lines~\ref{line:update-x} and \ref{line:update-z}). Consider an iteration of the inner while loop. If the algorithm executes line~\ref{line:fz-larger}, we have $\vec{x} = \vec{z}$ at the end of the iteration. Therefore we may assume that the algorithm does not execute line~\ref{line:fz-larger}. Let $\vec{x}'$ and $\vec{z}'$ be the updated vectors after performing the updates on line~\ref{line:update-x} and line~\ref{line:update-z}, respectively. Let $\vec{x}$ and $\vec{z}$ denote the vectors right before the update. By the induction hypothesis, we have $\vec{0} \leq \vec{x} \leq \vec{z}$.

For each coordinate $i \in [n]$, we have:
\begin{align*}
  \vec{x}_i' &=
  \begin{cases}
    \vec{x}_i + \eta (1 - \vec{x}_i) = \eta + (1 - \eta) \vec{x}_i &\text{if } i \in T(\eta)\\
    \vec{x}_i &\text{otherwise}
  \end{cases}\\
  \vec{z}_i' &=
  \begin{cases}
    \vec{z}_i + \eta (1 - \vec{z}_i) = \eta + (1 - \eta) \vec{z}_i &\text{if } i \in S\\
    \vec{z}_i &\text{otherwise}
  \end{cases}
\end{align*}
Since $\vec{0} \leq \vec{x} \leq \vec{z}$, $1 - \eta \geq 0$, and $T(\eta) \subseteq S$, we have $\vec{0} \leq \vec{x}' \leq \vec{z}'$, as needed.
\end{proof}

\begin{lemma}
\label{lem:nonneg-grad-seta}
Consider an iteration of the inner while loop. Let $\vec{x}$ and $\vec{z}$ be the respective vectors before the updates on lines~\ref{line:update-x}--\ref{line:fz-larger}, and let $\vec{x}'$ and $\vec{z}' = \vec{z}(\eta)$ be the respective vectors after the updates. For each coordinate $i \in S(\eta)$, we have $\nabla_i f(\vec{x}') \geq \nabla_i f(\vec{z}') \geq \frac{v}{1 - \vec{z}'_i} > 0$.
\end{lemma}
\begin{proof}
Note that the update rule on line~\ref{line:update-z} sets $\vec{z}' = \vec{z}(\eta)$. Since $\vec{x}' \leq \vec{z}'$ (Lemma~\ref{lem:xz}), DR-submodularity implies that $\nabla f(\vec{x}') \geq \nabla f(\vec{z}')$. Let $i \in S(\eta)$. By the definition of $S(\eta)$, we have $\vec{z}_i(\eta) \leq 1 - (1 - \eps)^j < 1$ and $(1 - \vec{z}_i(\eta)) \nabla_i f(\vec{z}(\eta)) \geq v > 0$, which implies that $\nabla_i f(\vec{z}(\eta)) \geq \frac{v}{1 - \vec{z}_i(\eta)} > 0$.
\end{proof}

\begin{lemma}
\label{lem:monotonicity}
Consider the vectors and sets defined on line~\ref{line:defs}. For all $\eta$ and $\eta'$ such that $0 \leq \eta \leq \eta' \leq \eps$, we have:
\begin{itemize}
\item[(1)] $\vec{z}(\eta) \leq \vec{z}(\eta')$,
\item[(2)] $S(\eta) \supseteq S(\eta')$.
\item[(3)] $T(\eta) \supseteq T(\eta')$.
\end{itemize}
\end{lemma}
\begin{proof}
\begin{itemize}
\item[(1)] For every $i \notin S$, we have $\vec{z}_i(\eta) = \vec{z}_i(\eta') = \vec{z}_i$. For every $i \in S$, we have:
\[ \vec{z}_i(\eta) \overset{(a)}{=} \vec{z}_i + \eta (1 - \vec{z}_i) \overset{(b)}{\leq} \vec{z}_i + \eta' (1 - \vec{z}_i) \overset{(c)}{=} \vec{z}_i(\eta')\]
where (a) and (c) are due to $i \in S$, (b) is due to $\eta \leq \eta'$ and $1 -\vec{z}_i \geq 0$ (since $i \in S$, $\vec{z}_i \leq 1 - (1 - \eps)^j \leq 1$).
\item[(2)] Let $i \in S(\eta')$. By (1) and DR-submodularity, we have $\nabla f(\vec{z}(\eta)) \geq \nabla f(\vec{z}(\eta'))$. Since $i \in S(\eta')$, we have $1 - z_i(\eta') \geq 0$ (since $z_i(\eta') \leq 1 - (1 - \eps)^j \leq 1$), and $\nabla_i f(\vec{z}(\eta')) \geq 0$ (since $1 - z_i(\eta') \geq 0$ and $(1 - z_i(\eta')) \nabla_i f(\vec{z}(\eta')) \geq v > 0$). Therefore
\[ \vec{g}_i(\eta) \overset{(a)}{=} (1 - \eta) (1 - \vec{z}_i) \nabla_i f(\vec{z}(\eta)) \overset{(b)}{\geq} (1 - \eta') (1 - \vec{z}_i) \nabla_i f(\vec{z}(\eta')) \overset{(c)}{=} \vec{g}_i(\eta') \overset{(d)}{\geq} v\]
where (a) and (c) are due to $i \in S$; (b) is due to $\eta \leq \eta'$, $1 - \vec{z}_i \geq 0$, and $\nabla_i f(\vec{z}(\eta)) \geq \nabla_i f(\vec{z}(\eta')) \geq 0$; $(d)$ is due to $i \in S(\eta')$.
\item[(3)] Let $i \in T(\eta')$. By (1) and DR-submodularity, we have $\nabla f(\vec{z}(\eta)) \geq \nabla f(\vec{z}(\eta'))$. Since $i \in T(\eta')$, we have $1 - z_i(\eta') > 0$ and $\nabla_i f(\vec{z}(\eta')) > 0$. Therefore
\[ \vec{g}_i(\eta) \overset{(a)}{=} (1 - \eta) (1 - \vec{z}_i) \nabla_i f(\vec{z}(\eta)) \overset{(b)}{\geq} (1 - \eta') (1 - \vec{z}_i) \nabla_i f(\vec{z}(\eta')) \overset{(c)}{=} \vec{g}_i(\eta') \overset{(d)}{>} 0\]
where (a) and (c) are due to $i \in S$; (b) is due to $\eta \leq \eta'$, $1 - \vec{z}_i \geq 0$, and $\nabla_i f(\vec{z}(\eta)) \geq \nabla_i f(\vec{z}(\eta')) \geq 0$; $(d)$ is due to $i \in T(\eta')$.
\end{itemize}
\end{proof}

\section{Approximate step sizes}
\label{app:approx-step-size}

\begin{algorithm}
\caption{Algorithm for $\max_{\vec{x} \in [0, 1]^n \colon \|\vec{x}\|_1 \leq k} f(\vec{x})$, where $f$ is a non-negative DR-submodular function.} 
\label{alg:non-monotone}
\begin{algorithmic}[1]
\State $M: f(\xopt) \leq M \leq (1 + \eps) f(\xopt)$
\State $\x \gets \zero$
\State $\z \gets \zero$
\For{$j = 1$ to $1/\eps$}
  \State \LineComment{Start of phase $j$}
  \State $\xstart \gets \x$
  \State $\zstart \gets \z$
  \State $\vstart \gets \frac{1}{k}(((1 - \eps)^j - 2\eps) M - f(\x))$ \label{line:vstart-approx}
  \State $v \gets \vstart$
  \While{$v > \eps \vstart$ and $\|\z\|_1 < \eps jk$}
    \State $\g = (\one - \z) \circ \nabla f(\z)$
    \State $S = \{i \in [n] \colon \g_i \geq v \text{ and } \z_i \leq 1 - (1 - \eps)^j \text{ and } \z_i - (\zstart)_i < \eps (1 - (\zstart)_i)\}$ \label{line:S-approx}
    \If{$S = \emptyset$}
      \State $v \gets (1 - \eps) v$ \label{line:update-v-approx}
    \Else
      \State For a given $\eta \in [0, \eps^2]$, we define:
      \begin{align*}
        \zet &= \z + \eta (\one - \z) \circ \one_S\\
        \get &= (\one - \zet) \circ \nabla f(\zet)\\
        \Set &= \{i \in S \colon \get_i \geq v \}\\
        \Tet &= \{i \in S \colon \get_i > 0 \}
      \end{align*}
      \label{line:defs-approx}
      \State \LineComment{$\delta \leq \eps / N$, where $N$ is the total number of iterations of the algorithm}
      \State Let $\delta = \Theta\left(\frac{\eps^4}{\log(n) \log(1/\eps)} \right)$ \label{line:delta-approx}
      \State \LineComment{Let $\eta^*_1$ be the maximum $\eta \in [0, \eps^2]$ such that $|\Set| \geq (1 - \eps) |S|$} \label{line:eta1-opt}
      \State Using $t$-ary search, find $\eta_1 \in [0, \eps^2]$ such that $\eta^*_1 \leq \eta_1 \leq \eta^*_1 + \delta$
      \State \LineComment{$\eta_2 = \min\left\{\eps^2,  \frac{\eps jk - \|\z\|_1}{|S| - \|\z \circ \one_S\|_1} \right\}$} \label{line:eta1-approx}
      \State Let $\eta_2$ be the maximum $\eta \in [0, \eps^2]$ such that $\|\zet\|_1 \leq \eps jk$ \label{line:eta2-approx}
      \State $\eta \gets \min\{\eta_1, \eta_2\}$ \label{line:eta-approx}
      \State $\x \gets \x + \eta (\one - \x) \circ \one_{T(\eta - \delta)}$ \label{line:update-x-approx}
      \State $\z \gets \z + \eta (\one - \z) \circ \one_S$ \label{line:update-z-approx}
      \If{$f(\z) > f(\x)$}
        \State $\x \gets \z$ \label{line:fz-larger-approx}
      \EndIf
    \EndIf
  \EndWhile
\EndFor
\State \Return $\x$
\end{algorithmic}
\end{algorithm}

In this section, we show how to extend the idealized algorithm (Algorithm~\ref{alg:non-monotone-ideal}) and its analysis. In order to obtain an efficient algorithm, we find the step size $\eta_1$ approximately using $t$-ary search, as described below. The modified algorithm is given in Algorithm~\ref{alg:non-monotone}. On line~\ref{line:delta-approx} of Algorithm~\ref{alg:non-monotone}, the $\Theta$ notation hides a sufficiently small constant so that $\delta \leq \eps / N$, where $N$ is the total number of iterations of the algorithm (as we discuss later in this section, the analysis of the number of iterations given in Theorem~\ref{thm:iters} still holds and thus $N = O(\log(n) \log(1/\eps)/\eps^3)$.) 

{\bf Finding $\eta_1$ on line~\ref{line:eta1-approx}.} As in the description of the algorithm, we let $\eta^*_1$ be the maximum $\eta \in [0, \eps^2]$ such that $|\Set| \geq (1 - \eps) |S|$ and we let $\delta$ be the value on line~\ref{line:delta-approx}. As shown in Lemma~\ref{lem:monotonicity}, for every $\eta \leq \eta'$, we have $S(\eta) \supseteq S(\eta')$, and thus $|S(\eta)|$ is non-increasing as a function of $\eta$. Note that $S(0) = S$ and thus $|S(0)| \geq (1 - \eps) |S|$. We first check whether $|S(\eps^2)| \geq (1 - \eps) |S|$; if so, we have $\eta^*_1 = \eps^2$ and we return $\eta_1 = \eps^2$. Therefore we may assume that $|S(\eps^2)| < (1 - \eps) |S|$ and thus $\eta^*_1 \in [0, \eps^2)$. Starting with the interval $[0, \eps^2]$, we perform $t$-ary search, and we stop once we reach an interval $[a, b]$ of length at most $\delta$. We return $\eta_1 = b$. Note that we have $\eta^*_1 \leq \eta_1 \leq \eta^*_1 + \delta$. 

The arity of the $t$-ary search gives us different trade-offs between the number of parallel rounds and the total running time. The $t$-ary search takes $\log_t(\eps^2 / \delta)$ parallel rounds and $t \log_t (\eps^2 / \delta)$ evaluations of $f$ and $\nabla f$. If we use binary search ($t = 2$), the number of rounds is $\log_2(\eps^2 / \delta) = O(\log\log{n} + \log(1/\eps))$ and the number of evaluations of $f$ and $\nabla f$ is also $O(\log\log{n} + \log(1/\eps))$. If we take $t = \Theta(\log{n}/\eps)$, the number of rounds is $O(1)$ and the number of evaluations of $f$ and $\nabla f$ is $O(\log{n}/\eps)$.  

Next, we show how to extend the analysis given in Sections~\ref{sec:approx} and \ref{sec:iters}. We first note that the upper bound on the total number of iterations given in Theorem~\ref{thm:iters} still holds, since we have $\eta \geq \eta^* := \min\{\eta^*_1, \eta_2\}$ and $T(\eta - \delta) \supseteq T(\eta^*)$. Therefore it only remains to show that the approximate search only introduces an overall $O(\eps)$ additive error in the approximation guarantee. Since $\delta \leq \eps / N$, where $N$ is the total number of iterations of the algorithm, it suffices to show that the error is $O(\delta) f(\xopt)$ in each iteration.

We start with the following lemma:

\begin{lemma}
\label{lem:fn-value-drop}
Let $\alpha, \beta \in \R$ and $\u, \v \in \R^n$. Suppose that $0 \leq \beta \leq \alpha \leq 1$, $\u \leq \v \leq \alpha \one$, and $\v - \u  \leq \beta \one$. Then $f(\u) - f(\v) \leq \frac{\beta}{\alpha} f(\u)$.
\end{lemma}
\begin{proof}
For $t \in [0, 1]$, let $\w(t) := \u + (\v - \u) t$. The conditions in the lemma statement ensure that $\w(\alpha/\beta) \leq \one$. Using that $f$ is concave in non-negative directions and $f$ is non-negative, we obtain:
\[ f(\v) = f(\w(1)) \geq \left(1 - \frac{\beta}{\alpha} \right) f(\w(0)) + \frac{\beta}{\alpha} f(\w(\alpha / \beta)) \geq \left(1 - \frac{\beta}{\alpha} \right) f(\w(0)) = \left(1 - \frac{\beta}{\alpha} \right) f(\u)\]
The lemma now follows by rearranging the above inequality.
\end{proof}

We now fix an iteration of the algorithm (an iteration of the inner while loop) that updates $\x$ and $\z$ on lines~\ref{line:update-x-approx}--\ref{line:fz-larger-approx}. Let $\x, \z$ denote the vectors right before the update on lines~\ref{line:update-x-approx}--\ref{line:fz-larger-approx}. We define:
\begin{align*}
\eta^* & := \min\{\eta^*_1, \eta_2\}\\
\xp & := \x + \eta (\one - \x) \circ \one_{T(\eta - \delta)}\\
\a & := \x + (\eta - \delta) (\one - \x) \circ \one_{T(\eta - \delta)}\\
\b & := \x + \eta^* (\one - \x) \circ \one_{T(\eta^*)} + (\eta - \delta) (\one - \x) \circ \one_{T(\eta - \delta) \setminus T(\eta^*)}
\end{align*}
Note that we have $\eta - \delta \leq \eta^* \leq \eta$ and thus it follows from Lemma~\ref{lem:monotonicity} that $T(\eta - \delta) \supseteq T(\eta^*) \supseteq T(\eta)$.

We start by applying Lemma~\ref{lem:fn-value-drop}. Let $\u = \a$ and $\v = \xp$. We have $\xp \leq (1 - (1 - \eps)^j + \eps^2) \one$, and thus we can take $\alpha = 1 - (1 - \eps)^j + \eps^2 \leq 1 - (1 - \eps)^{1/\eps} + \eps^2 \approx \frac{1}{e} + \eps^2$. We have $\a \leq \xp$ and $\xp - \a \leq \delta \one$, and thus we can take $\beta = \delta$. It follows from Lemma~\ref{lem:fn-value-drop} that:
\[ f(\xp) - f(\a) \geq - O(\delta) f(\a) \geq - O(\delta) f(\xopt), \]
where in the second inequality we have used that $\a$ is feasible.

Next, we have:
\begin{align*}
f(\a) - f(\x) 
&\overset{(a)}{\geq} \langle \nabla f(\a), \a - \x \rangle\\
&= \langle \nabla f(\a), (\eta - \delta)  (\one - \x) \circ \one_{T(\eta - \delta)} \rangle\\
&\overset{(b)}{\geq} \langle \nabla f(\a), (\eta - \delta) (\one - \x) \circ \one_{T(\eta^*)} \rangle\\
&\overset{(c)}{\geq} \langle \nabla f(\z(\eta^*)), (\eta - \delta) (\one - \x) \circ \one_{T(\eta^*)} \rangle\\
&\overset{(d)}{\geq} \langle \g(\eta^*), (\eta - \delta) \one_{T(\eta^*)} \rangle
\end{align*}
In (a), we used that $f$ is concave in non-negative directions and $\a \geq \x$. We can show (b) as follows. As noted earlier, $T(\eta^*) \subseteq T(\eta - \delta)$. Since $\a \leq \z(\eta - \delta)$, we have $\nabla f(\a) \geq \nabla f(\z(\eta - \delta))$ by DR-submodularity, and thus $\nabla f(\a)$ is non-negative on the coordinates in $T(\eta - \delta)$. In (c), we have used that $\a \leq \z(\eta - \delta) \leq \z(\eta^*)$ and thus $\nabla f(\a) \geq \nabla f(\z(\eta^*))$ by DR-submodularity. In (d), we used that $\nabla f(\z(\eta^*))$ is non-negative on the coordinates in $T(\eta^*)$ and $\one - \x \geq \one - \z(\eta^*) \geq \zero$.

Similarly, we have:
\begin{align*}
f(\b) - f(\a)
&\overset{(a)}{\geq} \langle \nabla f(\b), \b - \a \rangle\\
&\overset{(b)}{=} \langle \nabla f(\b), (\eta^* - \eta + \delta) (\one - \x) \circ \one_{T(\eta^*)} \rangle\\
&\overset{(c)}{\geq} \langle \g(\eta^*), (\eta^* - \eta + \delta) \one_{T(\eta^*)} \rangle
\end{align*}
In (a), we used that $f$ is concave in non-negative directions and $\b \geq \a$. In (b), we used that $T(\eta - \delta) \supseteq T(\eta^*)$. We can show (c) as follows. Since $\eta - \delta \leq \eta^*$ and $T(\eta - \delta) \subseteq S$, we have $\b \leq \z(\eta^*)$. Thus $\nabla f(\b) \geq \nabla f(\z(\eta^*))$ by DR-submodularity and $\nabla f(\b)$ is non-negative on the coordinates of $T(\eta^*)$.

By combining the inequalities above, we obtain:
\[ f(\xp) - f(\x) \geq \langle \g(\eta^*), \eta^* \one_{T(\eta^*)} \rangle - O(\delta) f(\xopt) \]
Thus we see that the gain obtained in the iteration is the one required by the proof of Theorem~\ref{thm:phase-gain} apart from the additive loss of $O(\delta) f(\xopt)$. By propagating the additive loss through the proof of Theorem~\ref{thm:phase-gain}, we obtain a total loss of $O(\delta N) f(\xopt)$, where $N$ is the total number of iterations. As noted above, $O(\delta N) = O(\eps)$, as needed.

\section{DR-submodular algorithms}
\label{app:algorithms}

In this section, we give the pseudocode of the sequential and parallel algorithms evaluated in our experiments. The sequential algorithm we used is the continuous greedy algorithm shown in Algorithm~\ref{alg:cg}. The algorithm is a variant of the measured continuous greedy algorithm that was studied in previous works~\cite{ChekuriJV15,bian2017continuous}. This variant obtains higher function value in practice, since it allows for the possibility of filling up more of the available budget, and this is what we observed in our experiments as well. The state of the art parallel algorithm for non-monotone DR-submodular maximization subject to a cardinality constraint is the algorithm of \cite{ENV19}; Algorithm~\ref{alg:mwu} gives the pseudocode of this algorithm specialized to a single cardinality constraint.

\begin{algorithm}
\caption{A variant of the measured continuous greedy algorithm for $\max_{\vec{x} \in [0, 1]^n \colon \|\vec{x}\|_1 \leq k} f(\vec{x})$, where $f$ is a non-negative DR-submodular function.} 
\label{alg:cg}
\begin{algorithmic}[1]
\State $\x \gets \zero$
\State \LineComment{In our experiments, we used $\eta = \eps / n$}
\State $\eta \gets \eps / n^3$
\State $T \gets 1 / \eta$
\For{$t = 1$ to $T$}
  \State $\d \gets \arg\max_{\z \in [0, 1]^n \colon \z \leq \one - \x, \|\z\|_1 \leq k} \langle \nabla f(\x), \z \rangle$
  \State $\x \gets \x + \eta \d$
\EndFor
\State \Return $\x$
\end{algorithmic}
\end{algorithm}

\begin{algorithm}
\caption{The algorithm of \cite{ENV19} specialized to a single cardinality constraint. The algorithm solves the problem $\max_{\vec{x} \in [0, 1]^n \colon \|\vec{x}\|_1 \leq k} f(\vec{x})$, where $f$ is a non-negative DR-submodular function. The algorithm takes as input a target value $M$ such that $f(\xopt) \leq M \leq (1 + \eps) f(\xopt)$.} 
\label{alg:mwu}
\begin{algorithmic}[1]
\State $\eta \gets \frac{\eps}{2 \log(n + 1)}$
\State $\x \gets \frac{\eps}{n} \one$
\State $\z \gets \x$
\State \LineComment{MWU weights for the $(n + 1)$ constraints $\z_i \leq 1$ for all $i \in [n]$ and $\frac{1}{k} \langle \z, \one \rangle \leq 1$}
\State $\w_i \gets \exp(\z_i / \eta)$ for all $i \in [n]$
\State $\w_{n + 1} \gets \exp(\|\z\|_1 / (\eta k))$
\State $t \gets \eta \ln(\|\w\|_1)$
\While{$t < 1 - \eps$}
\State $\lambda \gets M \cdot (e^{-t} - 2\eps) - f(\x)$
\State $\c \gets (\one - \x) \circ \nabla f((1 + \eta) \x) \vee \zero$
\State $\m_i \gets \left( 1 - \lambda \cdot \frac{1}{\c_i} \cdot \frac{1}{\|\w\|_1} \left(\w_i + \frac{1}{k}\w_{n + 1}  \right) \right) \vee 0$ for all $i \in [n]$ with $\c_i \neq 0$, and $\m_i = 0$ if $\c_i = 0$
\State $\d \gets \eta \x \circ \m$
\If{$\d = \zero$}
  \State break
\EndIf
\State $\x \gets \x + \d \circ (\one - \x)$
\State $\z \gets \z + \d$
\State \LineComment{Update the weights}
\State $\w_i \gets \exp(\z_i / \eta)$ for all $i \in [n]$
\State $\w_{n + 1} \gets \exp(\|\z\|_1 / (\eta k))$
\EndWhile
\State \Return $\x$
\end{algorithmic}
\end{algorithm}
\end{document}